\documentclass{article}
\usepackage[utf8]{inputenc}
\usepackage{graphicx}
\usepackage{float}
\usepackage{dsfont,complexity,xcolor}
\usepackage{url}
\usepackage{amsmath,amsfonts,amssymb}
\usepackage{multirow}

\usepackage{amsthm}

\usepackage{enumitem}

  \usepackage{tikz}
  \usetikzlibrary{positioning,fit,calc,shapes,backgrounds}
  \usetikzlibrary{matrix}

\DeclareMathOperator{\T}{\mathcal{T}}
\usepackage{xspace}

\usepackage{soul,color}
\newcommand{\todo}[2]{\hl{#1 : #2}}

\newcommand{\probgen}{$k$ $X$-spanning $Y$-disjoint Branching\xspace}
\newcommand{\probsinglesource}{$k$ Single Source $X$-spanning $Y$-disjoint Branching\xspace}
\newcommand{\probgenarg}[1]{$k+1$ $X$-spanning $Y$-disjoint Branching\xspace}

\newcommand{\NAESAT}{\textsc{NAE$+3$-SAT}}

\newtheorem{theorem}{Theorem}

\newtheorem{lemma}{Lemma}

\theoremstyle{definition}
\newtheorem{problem}{Problem}
\newtheorem{proposition}{Proposition}

\title{Edge-Disjoint Branchings in Temporal Graphs\thanks{Emails of authors: \{\texttt{campos,raul.lopes}\}\texttt{@lia.ufc.br}, \texttt{andrea.marino.unifi.it}, \texttt{anasilva@mat.ufc.br}.}}
\author{Victor Campos\thanks{Departamento de Computa\c{c}\~ao, Universidade Federal do Cear\'{a}, Fortaleza, CE, Brazil.} \and Raul Lopes\footnotemark[2] \and Andrea Marino\thanks{Dipartimento di Sistemi, Informatica, Applicazioni, Universit\`{a} degli Studi di Firenze, Firenze, Italy.} \and Ana Silva\thanks{Departamento de Matem\'{a}tica, Universidade Federal do Cear\'{a}, Fortaleza, CE, Brazil.}}

\begin{document}
\maketitle

\begin{abstract}
    A temporal digraph ${\cal G}$ is a triple $(G, \gamma, \lambda)$ where $G$ is a digraph, $\gamma$ is a function on $V(G)$ that tells us the timestamps when a vertex is active, and $\lambda$ is a function on $E(G)$ that tells for each $uv\in E(G)$ when $u$ and $v$ are linked.
    Given a static digraph $G$, and a subset $R\subseteq V(G)$, a spanning branching with root $R$ is a subdigraph of $G$ that has exactly one path from $R$ to each $v\in V(G)$. In this paper, we consider the temporal version of Edmonds' classical result about the problem of finding $k$ edge-disjoint spanning branchings respectively rooted at given $R_1,\cdots,R_k$. We introduce and investigate different definitions of spanning branchings, and of edge-disjointness in the context of temporal graphs.
    A branching ${\cal B}$ is vertex-spanning if the root is able to reach each vertex $v$ of $G$ at some time where $v$ is active, while it is temporal-spanning if $v$ can be reached from the root at every time where $v$ is active. On the other hand, two branchings ${\cal B}_1$ and ${\cal B}_2$ are edge-disjoint if they do not use the same edge of $G$, and are temporal-edge-disjoint if they can use the same edge of $G$ but at different times. This lead us to four definitions of disjoint spanning branchings and we prove that, unlike the static case, only one of these can be computed in polynomial time, namely the temporal-edge-disjoint temporal-spanning branchings problem, while the other versions are $\NP$-complete, even under very strict assumptions.
\end{abstract}

\section{Introduction}

A \emph{temporal digraph} is a digraph that exists in time, meaning that, given a digraph $G$, vertices might be active or inactive at a certain timestamp, and edges might have a delay, leaving a vertex at a timestamp, but only arriving later.
In this paper we deal with \emph{disjoint spanning branchings} in temporal graphs, which are well-understood structures in static graphs. Given a digraph $G$, and a subset $R\subseteq V(G)$, we say that $H\subseteq G$ is a \emph{spanning branching} of $G$ with root $R$ if $V(H) = V(G)$, and $H$ contains exactly one path between some $r\in R$ and $u$, for each $u\in V(G)$. Given subsets $R_1,\cdots,R_k$, a classical result by Edmonds~\cite{edmonds1973edge} gives a necessary and sufficient condition for the existence of $k$ edge-disjoint branchings with roots $R_1,\cdots,R_k$, respectively. His result also gives a polynomial algorithm that constructs these branchings.  
%

When translating concepts to temporal graphs, it is often the case that theorems coming from static graph theory can hold or not depending on the adopted definition. Indeed, in~\cite{KKK00} the authors give an example where Edmond's result on branchings does not hold on the temporal context. However, as we will see later, their concept is just one of many possible definitions, and that in fact there is even one case where polinomiality holds.

Another example of such behavior is the validity of Menger's Theorem. It has been shown that the edge version of Menger's Theorem holds~\cite{B.96}, even if one considers weights on the edges~\cite{ACGKS.19}. However, the vertex version of Menger's Theorem holds or not, depending on how one interprets what a cut should be. If a cut is understood as a subset of $V(G)$, then Menger's Theorem does not hold~\cite{B.96,KKK00}; and if it is understood as a subset of the appearances of vertices in time (alternatively, a cut can be seen as deactivating vertices at some timestamps), then Menger's Theorem holds~\cite{MMS.19}.


\paragraph{Our contribution.} Given a temporal digraph ${\cal G}$ with base static digraph $G$, and subsets of vertices $R_1,\cdots,R_k$ in time, i.e. sets of pairs $(u,t)$ where $u$ is a vertex of $G$ and $t$ a timestamp, here we investigate the many variations of finding (pairly) disjoint spanning branchings with roots $R_1,\cdots,R_k$. Spanning can mean that one wants to pass by at least one appearance of each  $u\in V(G)$ (called \emph{vertex spanning}), or by all appearances of each $u\in V(G)$ (called \emph{temporal spanning}).
Similarly, edge-disjoint can have different interpretations, as it can refer to edges of $G$ or to the appearances of these edges in ${\cal G}$. We say that two branchings are \emph{edge-disjoint} if they do not share any edge of $G$, and that they are \emph{temporal-edge-disjoint} (or \emph{t-edge-disjoint} for short) if they do not share any appearance of an edge of $G$ in  ${\cal G}$.
We found that the only case in which this problem is polynomial (as its static counterpart) is when we want t-edge-disjoint temporal-spanning branchings. We also found that if vertices are \emph{eternal} (this is the more popular case where vertices are always active), the problem is polynomial for temporal-spanning branchings and \NP-complete otherwise. Our results are summarized in Table \ref{tab:results} and detailed in the following main theorem.

\begin{theorem}
Let ${\cal G}$ be a temporal digraph with base digraph $G$, and consider subsets of vertices in time, $R_1,\cdots,R_k$. The problem of finding $k$ branchings rooted at $R_1,\cdots,R_k$ is: 
\begin{enumerate}
    \item Polynomial for t-edge-disjoint temporal-spanning, \label{item:one}
    \item $\NP$-complete for edge-disjoint temporal-spanning even if: $G$ is a star, and each snapshot has constant size; or if ${\cal G}$ has lifetime~3.
    And if vertices are eternal or ${\cal G}$ has lifetime~2, then edge-disjoint temporal-spanning becomes polynomial. \label{item:two}
    
    \item $\NP$-complete for edge-disjoint vertex-spanning even if $G$ is a DAG, the lifetime of ${\cal G}$ is~2, and vertices are eternal.\label{item:three}
    \item $\NP$-complete for t-edge-disjoint vertex-spanning  even if $G$ is a DAG, the lifetime of ${\cal G}$ is~2, and vertices are eternal.\label{item:four}
\end{enumerate}
\label{thm:main}
\end{theorem}

As said before, Edmonds' condition is the characterization behind the polynomial algorithm for finding $k$ edge disjoint spanning branchings in static digraphs. Because of our \NP-completeness results, it is worth remarking that, unless \P=\NP, any such characterization for the \NP-complete cases in temporal digraphs should be checkable in exponential time, differently from the one provided by Edmonds. 

\sloppy Finally, our reductions further implies that, in the case of edge-disjoint temporal-spanning, even if the base graph $G$ is a star, the problem cannot be solved by an algorithm running in time $O^*(2^{o(\T)})$ unless ETH fails, where $\T$ is the lifetime of ${\cal G}$. Moreover, in the vertex-spanning variations, also the problem cannot be solved in $O^*(2^{o(n+m)})$ under the same assumption, where $n$ and $m$ are respectively the number of nodes and edges of the base digraph of ${\cal G}$.

\begin{table}[t]
\small
    \centering
    \begin{tabular}{|c|c|c|c|c|}
    \cline{2-5}
    \multicolumn{1}{c|}{} &  \multicolumn{2}{c|}{\textsc{not eternal vertices}} & \multicolumn{2}{c|}{\textsc{eternal vertices}}\\
    \cline{2-5}
     \multicolumn{1}{c|}{} &  \textsc{edge-disjoint} & \textsc{t-edge-disjoint} & \textsc{edge-disjoint} & \textsc{t-edge-disjoint}\\
    \hline
    \textsc{temporal-} & \multirow{2}{*}{Poly} & \multirow{2}{*}{\NP-c} & \multirow{2}{*}{Poly} & \multirow{2}{*}{Poly}\\
    \textsc{spanning} & & & & \\
    \hline
    \textsc{vertex-} & \multirow{2}{*}{\NP-c} & \multirow{2}{*}{\NP-c} & \multirow{2}{*}{\NP-c} & \multirow{2}{*}{\NP-c}\\
    \textsc{spanning} & & & & \\
    \hline
    \end{tabular}
    
    \caption{Our results. Vertices are eternal if they are always active.}
    \label{tab:results}
\end{table}

\paragraph{Related Work:} While it is easy to imagine a variety of graph problems that can profit from considering changes in time, it is hard to pin-point when the study of temporal graphs and similar structures began. Nevertheless, in the last decade or so, it has attracted a lot of attention from the community, with a considerable number of papers being published in the field (we refer the reader to the surveys~\cite{M15,LVM.18}). It is worth observing, however, that different names have been used for denoting temporal networks (even though the basic notion was almost the same), such as, for example, dynamic networks~\cite{BhadraF03}, time-varying graphs~\cite{CFQS12}, evolving networks~\cite{BorgnatFGMRS07}, and link streams~\cite{LVM.18}. Also, many works consider a temporal graph ${\cal G}$ as having vertices that are always active, and edges have the same starting and ending time~\cite{KKK00,MMS.19,ACGKS.19,XFJ03,SQFCA11}. While models where edges that have a delay are more common~\cite{CFQS12,WCHLX14}, models where nodes can be inactive have already been considered in~\cite{CFQS12,LVM.18}. 

A path in temporal graphs is generally understood as a sequence of edges respecting time, i.e. the arrival time in each vertex of the path must be lower than the departing time of the next edge taken. In this context, a number of metrics can be related to a path, such as earliest arrival time, latest departure time, minimum number of temporal edges, and minimum traveling time \cite{WCHLX14}. When vertices can be inactive, we have to further ensure that, when waiting for the next edge on a certain vertex, it must remain active in the waiting period~\cite{CFQS12}. In this scenario, the definitions of reachability and connectivity change accordingly, and it is natural to ask how well-known structures and results from static graph theory change taking into account the temporal constraint.

Temporal definitions of trees~\cite{LVM.18} and (minimum) spanning trees~\cite{HFL2015}, which are related to our definition of branching, have been proposed and investigated, and usually consists of ensuring that the root-to-node path in the tree is a valid temporal path. Analogously, temporal cuts from a vertex $s$ to $t$ aim to break any temporal path from $s$ to $t$ and can be related to extending the max-flow min-cut Theorem to temporal graphs~\cite{ACGKS.19}. And as we have already mentioned, different conclusions have been made about a temporal version of Menger's Theorem depending on the adopted translation in terms of temporal graphs~\cite{B.96,KKK00,MMS.19}. 

Edmond's Theorem on disjoint branchings is a classical theorem in graph theory, with many distinct existing proofs (e.g. Lov\'asz~\cite{Lovasz76}, Tarjan~\cite{Tarjan74}, and Fulkerson and Harding~\cite{FH76}), and has many interesting consequences on digraph theory (e.g., one can derive Menger's Theorem from it, characterize arc-connectivity~\cite{Shiloach79}, characterize branching cover~\cite{F79}, ensure integer decomposition of the polytope of branchings of size $k$~\cite{McD83}, etc). As far as we know, the only other time that Edmond's Theorem has been investigated on the temporal context has been in~\cite{KKK00}, where the authors give an example where the theorem does not hold. The definition used by them falls into our category of edge-disjoint vertex-spanning branchings, which we have proved to be $\NP$-complete even under very strict constraints. 

\paragraph{Structure of the paper.} The paper is organized as follows. In Section~\ref{sec:defs}, we formalize the definitions of spanning branchings and disjointness, also showing that having multiple roots in each of the $k$ branchings is computationally equivalent to having a single root for all of the $k$ branchings. In Section~\ref{sec:temporalSpanning}, we present the results about temporal-spanning branchings. In Section~\ref{sec:vertexSpanning} we present our results concerning vertex-spanning branchings. Finally, in Section~\ref{sec:conclusions}, we draw our conclusions and make some final remarks.

\section{The Temporal Disjoint Branchings Problems}\label{sec:defs}
This section is devoted to formally define the several concepts of temporal graphs and disjoint branchings we introduce in this paper. 
A temporal digraph ${\cal G}$ is a triple $(G, \gamma, \lambda)$ where $G$ is a digraph and $\gamma$ and $\lambda$ are functions on $V(G)$ and $E(G)$, respectively, that tell us when the vertices and the edges appear. More formally, for each $v\in V(G)$ we have $\gamma(v)\subseteq \mathbb{N}$, and for each edge $e\in E(G)$ we have $\lambda(e)\subseteq \mathbb{N}\times \mathbb{N}$. Also, if $(t, t') \in \lambda(uv)$, then $t\le t'$, $t\in \gamma(u)$ and $t'\in \gamma(v)$. Here, we consider only finite temporal digraphs, i.e., $\T=\max \bigcup_{v\in V(G)}\gamma(v)$ is defined and is called the \emph{lifetime of $\cal G$}. We call $G$ the \textit{base digraph of $\cal G$}
In what follows, unless said otherwise, we work on general digraphs, i.e., directions, loops and multiple edges are allowed.

In particular, if $\T$ is the lifetime of ${\cal G}=(G,\gamma,\lambda)$, if $\gamma(v)=[\T]$ for each $v\in V(G)$, and $t=t'$ for every $(t,t')\in \lambda(E(G))$, then the above definition corresponds to the definition of temporal graph given in \cite{KKK00} and many other works. The above definition also generalizes the definition of stream graph given in~\cite{LVM.18}, and of time-varying graphs given in~\cite{CFQS.11}.

The \textit{vertices} and \textit{edges} of ${\cal G}$ are the vertices and edges of $G$. We say that a vertex $v$ \textit{is active} at time $t$ if $t\in \gamma(v)$ and $v$ \textit{is active} from $t_1$ to $t_2$ if $v$ is active for every time $t$ with $t_1 \leq t \leq t_2$. 
%
The set $V_T$ of \textit{temporal vertices} is the set $\{(v,t)\mid v\in V(G)\ and\ t\in \gamma(v)\}$, and the set $E_T$ of \textit{temporal edges} is the set $\{(u,t)(v,t')\mid e=uv\in E(G)\ and\ (t,t')\in \lambda(e)\}$. 
Observe that a temporal digraph $\mathcal{G}=(G,\gamma,\lambda)$ can be also seen as a pair of digraphs $(G,G_T)$ where $G_T = (V_T,E_T)$. This is similar to what has been proposed in~\cite{CFQS.11} and~\cite{ACGKS.19}. We call the digraph $G_T$ the \emph{$(\gamma,\lambda)$-digraph of ${\cal G}$}.

Since in our more general case, also vertices appear and disappear, the definition of \emph{walk} must take into account that it is possible to wait only on vertices which are active, as formally defined next. Given temporal vertices $s_1,s_k\in V_T$, an \textit{$s_1,s_k$-temporal walk} in $(G,G_T)$ 
is a sequence of temporal vertices and temporal edges, $(s_1, \ldots, s_k)$, that either goes through a temporal edge, or stays on different copies of the same vertex of $G$. More formally: if $s_i$ is a temporal edge, then $s_{i-1}$ and $s_{i+1}$ are temporal vertices and $s_i$ goes from $s_{i-1}$ to $s_{i+1}$; and if $s_i$ and $s_{i+1}$ are temporal vertices, then $s_i = (v, t)$ and $s_{i+1} = (v, t+1)$ for some vertex $v$ and some time $t$. 
If such a walk exists, we say that $s_1$ \textit{reaches} $s_k$.

A temporal digraph ${\cal B}=(G',\gamma',\lambda')$ such that $G'\subseteq G$, $\gamma'\subseteq\gamma$ and $\lambda'\subseteq\lambda$ is called a \emph{temporal subdigraph of ${\cal G}$}.\footnote{Here, a function is seen as a set of ordered pairs, and the containment relation is the usual one for sets.} Let $R\subseteq V_T$; a temporal subdigraph ${\cal B}$ of ${\cal G}$ is a \emph{temporal-spanning branching} of ${\cal G}$ with root $R$ if ${\cal B}$ has a unique temporal walk from $R$ to every vertex in $V_T$, i.e. for any $(u,i)\in V_T$ there is exactly one temporal walk in ${\cal B}$ starting at some vertex $r\in R$ and arriving at $(u,i)$.
And ${\cal B}$ is a \emph{vertex-spanning branching} of ${\cal G}$ with root $R$ if ${\cal B}$ has exactly one temporal walk from $R$ to some vertex in $\{(u,i)\in V_T\}$ for every $u\in V(G)$.

Given two branchings ${\cal B}_1 = (G_1,\gamma_1,\lambda_1)$ and ${\cal B}_2 = (G_2,\gamma_2,\lambda_2)$ rooted at $R_1,R_2$, respectively, either both temporal-spanning or both vertex-spanning, we say that ${\cal B}_1$ and ${\cal B}_2$ are \emph{temporal-edge-disjoint} (or t-edge-disjoint for short) if they have no common temporal edges; more formally, if $\lambda_1(e)\cap \lambda_2(e) = \emptyset$ for every $e\in E(G)$.
And we say that ${\cal B}_1$ and ${\cal B}_2$ are \emph{edge-disjoint} if there is no edge $uv\in E(G)$ that has copies in both ${\cal B}_1$ and ${\cal B}_2$; more formally, $E(G_1)\cap E(G_2) = \emptyset$.
%

\begin{problem}[\textit{\probgen}]
Let $X\in\{\textrm{edge}, \textrm{t-edge}\}$, $Y\in\{\textrm{temporal}, \textrm{vertex}\}$, and $k$ be a fixed positive integer. Given a temporal digraph ${\cal G}$, and subsets of temporal vertices $R_1,\ldots,R_k\subseteq V_T$, find $k$ $X$-disjoint $Y$-spanning branchings ${\cal B}_1,\ldots, {\cal B}_k$ respectively with roots $R_1,\ldots,R_k$.
\label{prob:main}
\end{problem}

We introduce the following restriction of Problem~\ref{prob:main}, which corresponds to finding branchings that have a single root (also called out-arborescence).

\begin{problem}[\textit{\probsinglesource}]
\sloppy
Let $X\in\{\textrm{edge}, \textrm{t-edge}\}$, $Y\in\{\textrm{temporal}, \textrm{vertex}\}$, and $k$ be a fixed positive integer. Given a temporal digraph ${\cal G}$, and a temporal vertex $r\in V_T$, find $k$ $X$-disjoint $Y$-spanning branchings ${\cal B}_1,\ldots, {\cal B}_k$ each one with root $r$.
\label{prob:mainss}
\end{problem}


\begin{lemma}
Problem~\ref{prob:main} is computationally equivalent to Problem~\ref{prob:mainss}.
\label{lem:ssred}
\end{lemma}
\begin{proof}
Problem~\ref{prob:mainss} is clearly a restriction of  Problem~\ref{prob:main}. In the following we provide the reduction in the opposite direction, from the problem where each branching has a subset of $V_T$ as roots to the problem where each branching has a single same root. 
For this, for each $i\in [k]$ add a new vertex $r_i$ to $G$ adjacent to every $u\in V(G)$ such that $(u,t)\in R_i$, for some $t\in [\T]$. Then, make $\gamma(r_i)=\{0\}$, and for each $(u,t)\in R_i$, add $(0,t)$ to $\lambda(r_iu)$ (which is the same as adding the temporal edge $(r_i,0)(u,t)$ to ${\cal G}$). Moreover, add a vertex $r$ and make it adjacent to $\{r_1,\cdots,r_k\}$; also make $\gamma(r) = \{0\}$ and  $\lambda(rr_i) = \{(0,0)\}$ (which is the same as adding temporal edges $(r,0)(r_i,0)$ for every $i\in[k]$).

One can see that $k$ vertex-spanning (resp. temporal-spanning) branchings rooted at $r$ give $k$ vertex-spanning (resp. temporal-spanning) branchings rooted at $R_1,\cdots,R_k$, and vice-versa. The edge-disjointness, both for t-edge or edge-disjoint versions, clearly are not altered by adding the new temporal edges.
\end{proof}

The next easy proposition tells us that if finding $k$ disjoint spanning branchings is hard, for some fixed $k$, then so is finding $k+1$ of them.
\begin{proposition}\label{prop:ktok+1}
Let $X\in\{\textrm{temporal}, \textrm{vertex}\}$, $Y\in\{\textrm{edge}, \textrm{t-edge}\}$ and $k$ be a fixed positive integer. If Problem \probgen is $\NP$-complete, then the same holds for Problem \probgenarg{k+1}.
\end{proposition}
\begin{proof}
To reduce from $k$ to $k+1$, it suffices to add $R_{k+1} = V_T$ as entry. Surely the $(k+1)$-th branching has no temporal edges, which means that the other ones form a solution to the initial problem.
\end{proof}


\section{Temporal-Spanning Branchings}\label{sec:temporalSpanning}

This section is devoted to study Problem~\ref{prob:main} in the case where $Y$ is temporal, i.e. we aim to find $k$ $X$-disjoint temporal-spanning branchings, with $X\in\{\textrm{edge},\textrm{t-edge}\}$. We will hence prove Item~\ref{item:one} and Item~\ref{item:two} of Theorem~\ref{thm:main} respectively in Section~\ref{sec:polynomial} and in Section~\ref{sec:tempspanhard}.

\subsection{T-edge-disjoint Temporal-Spanning Branchings}
\label{sec:polynomial}


Let ${\cal G} = (G,\gamma,\lambda)$, and let $V_T,E_T$ be its set of temporal vertices and edges, respectively. Also, let $R_1,\cdots,R_k\subseteq V_T$, and $H = (V_T,E_T\cup E')$, where $E'$ contains $k$ copies of the edge $(u,t)(u,t+1)$ whenever $\{(u,t),(u,t+1)\}\subseteq V_T$. 
We prove that ${\cal G}$ has the desired branchings iff $H$ has $k$ edge-disjoint spanning branchings with roots $R_1,\cdots,R_k$. Then, Item~\ref{item:one} of Theorem~\ref{thm:main} follows by Edmonds' result~\cite{edmonds1973edge}.

\begin{lemma}
Let ${\cal G} = (G,\gamma,\lambda)$ be a temporal digraph, $R_1,\cdots,R_k\subseteq V_T$, and $H$ be constructed as above. 
Then, ${\cal G}$ has $k$ t-edge-disjoint temporal-spanning branchings rooted at $R_1,\cdots,R_k$ iff $H$ has $k$ edge-disjoint spanning branchings rooted at $R_1,\cdots,R_k$.
\end{lemma}
\begin{proof}
Let ${\cal B}_1,\cdots,{\cal B}_k$ be t-edge-disjoint temporal-spanning branchings rooted at $R_1,\cdots,R_k$, respectively. For each ${\cal B}_i$, let $B_i$ be a spanning subgraph of $H$ initially containing the temporal edges of ${\cal B}_i$; then for each $(u,t)\in V(B_i)$, if the only walk in ${\cal B}_i$ from $R_i$ to $(u,t)$ contains $(u,t)(u,t+1)$ as a subsequence, then add an unused copy of $(u,t)(u,t+1)\in $ to $B_i$. Because this walk is unique and cannot pass twice from timestamp $t$ to timestamp $t+1$, we get that at most $k$ copies are needed, and, hence, the produced branchings are edge-disjoint. The converse can be easily proved by deleting the edges in $E'$ from the solution to obtain the temporal subgraphs.
\end{proof}




\subsection{Edge-disjoint Temporal-Spanning Branchings}
\label{sec:tempspanhard}

In this section, we prove Item~\ref{item:two} of Theorem~\ref{thm:main}. For this, we first prove that the problem is \NP-complete, and then that it is polynomial when each vertex is active for a consecutive set of timestamps. This includes the popular case where vertices are assumed to be eternal, as well as the case where $T=2$. 

Theorem~\ref{thm:edgeDisjointTemporal} and Theorem~\ref{thm:theother} below detail our $\NP$-completeness results.

\begin{theorem}\label{thm:edgeDisjointTemporal}
Let $k\ge 2$ be a fixed integer, ${\cal G} = (G,\gamma,\lambda)$ be a temporal digraph, and $R_1,\ldots,R_k\subseteq V_T$. Deciding whether ${\cal G}$ has $k$ edge-disjoint temporal-spanning branchings rooted at $R_1,\cdots,R_k$ is $\NP$-complete even if ${\cal G}$ has lifetime~3.
\end{theorem}

We make a reduction from the $k$-\textsc{Weak Disjoint Paths} problem ($k$-\textsc{WDP}), where the input is a digraph $G$ and a set $I$ of $k$ pairs of vertices $\{(s_1, t_1), \ldots, (s_k, t_k)\}$ (called the \emph{requests}) of $V(G)$ and the goal is to find a collection of pairwise edge-disjoint paths $\{P_1, \ldots, P_k\}$ such that $P_i$ is a path from $s_i$ to $t_i$ in $G$, for $i \in \{1, \ldots, k\}$. 
This problem is $\NP$-complete for $k=2$~\cite{FORTUNE1980111} and \textsf{W}[1]-hard with parameter $k$ in DAGs~\cite{Slivkins.03}. 
We remark that we can always assume that $s_1,s_2$ are sources, $t_1,t_2$ are sinks, and all vertices in $\{s_1,s_2,t_1,t_2\}$ are distinct. 
If this is not the case for some $s_i$, for example, it suffices to add a new vertex $s'_i$ with an edge to $s_i$ and define $I' = (I - (s_i, t_i)) \cup (s'_i, t_i)$ as the new set of requests for the instance.

\begin{proof}[Proof of Theorem~\ref{thm:edgeDisjointTemporal}]
Let $(G, I)$ be an instance of $2$-\textsc{WDP} with $I = \{(s_1, t_1), (s_2, t_2)\}$, and define $W = \{s_1, t_1, s_2, t_2\}$.
Assume that $s_1,s_2$ are sources, $t_1,t_2$ are sinks, and all vertices in $W$ are distinct. We construct the temporal graph $\mathcal{G} = (G, \gamma, \lambda)$ with subsets $R_1,R_2$ such that ${\cal G}$ has $2$ edge-disjoint temporal-spanning branchings rooted at $R_1,R_2$ if and only if $(G,I)$ is a ``yes'' instance of $2$-\textsc{WDP}. The $\NP$-completeness for higher values of $k$ follows from Proposition~\ref{prop:ktok+1}.

In the constructed temporal graph, there are no temporal edges of the type $(u,t)(v,t')$ with $t\neq t'$. For this reason, it is easier to describe our temporal graph by describing, for each timestamp, what are the vertices and edges that are active. These are called \emph{snapshots} and consist of subgraphs of $G$ formed at each  timestamp. 

We let the first snapshot of ${\cal G}$ initially consist of $G-\{s_2,t_2\}$, and the third snapshot initially consist of $G-\{s_1,t_1\}$. Then, we add a new vertex $x$ to snapshot~1, and add the edges: $\{xv\mid v\in V(G)\setminus \{s_2,t_2\}\}\cup \{t_1v\mid v\in (V(G)\cup \{x\})\setminus \{s_1,s_2,t_2\}\}$. Similarly, we add a new vertex $y$ to snapshot~3, and add the edges: $\{yv\mid v\in V(G)\setminus \{s_1,t_1\}\}\cup \{t_2v\mid v\in (V(G)\cup \{y\})\setminus \{s_2,s_1,t_1\}\}$. Observe Figure~\ref{fig:2WDP}.

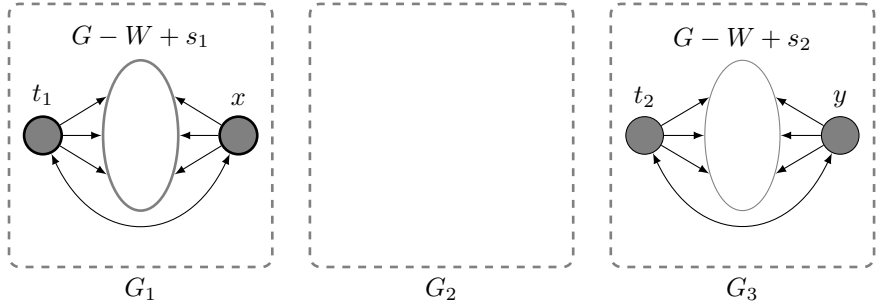
\begin{figure}[ht]
\begin{center}
  \begin{tikzpicture}[scale=1]
  \pgfsetlinewidth{1pt}
  \pgfdeclarelayer{bg}    
   \pgfsetlayers{bg,main}  

  \tikzset{vertex/.style={circle, minimum size=0.5cm, draw, inner sep=1pt, fill=black!50}}
  \tikzset{snapshot/.style ={draw=black!50, rounded corners, dashed, minimum height=8mm, minimum width=2.3cm} }
  \tikzset{subgraph/.style ={draw=black!50, circle, draw, minimum width=1cm, yscale=2, fill=white} }

    \node [label=-90:$G_1$, draw=black!50, rounded corners, dashed, minimum height=3.5cm, minimum width=3.5cm] at (0,0){};
    \node [label=-90:$G_2$, draw=black!50, rounded corners, dashed, minimum height=3.5cm, minimum width=3.5cm] at (4,0){};
    \node [label=-90:$G_3$, draw=black!50, rounded corners, dashed, minimum height=3.5cm, minimum width=3.5cm] at (8,0){};

    \node [subgraph, label=90:$G-W+{s_1}$] (G1) at (0,0) {};
    \node[vertex, label=90:$t_1$] (t1) at (-1.3,0) {};
    \node[vertex, label=90:$x$] (x) at (1.3,0) {};
    \coordinate (c1) at ($ (t1) + (0.7,-1.5) $);
    \coordinate (c2) at ($ (x) + (-0.7,-1.5) $);
    
    \begin{pgfonlayer}{bg}    
        \path[->,>=latex] (t1) edge (G1) (t1) edge(G1.-150) (t1) edge (G1.150);
        \path[->,>=latex] (x) edge (G1) (x) edge (G1.-30) (x) edge (G1.30);
        \draw[<->,>=latex]  (t1) .. controls (c1) and (c2) .. (x);
    \end{pgfonlayer}

    \node [subgraph, label=90:$G-W+{s_2}$, shift={(8,0)}] (G1) at (0,0) {};
    \node[vertex, label=90:$t_2$, shift={(8,0)}] (t1) at (-1.3,0) {};
    \node[vertex, label=90:$y$, shift={(8,0)}] (x) at (1.3,0) {};
    \coordinate (c1) at ($ (t1) + (0.7,-1.5) $);
    \coordinate (c2) at ($ (x) + (-0.7,-1.5) $);

    \begin{pgfonlayer}{bg}    
        \path[->,>=latex, shift={(8,0)}] (t1) edge (G1) (t1) edge (G1.-150) (t1) edge (G1.150);
        \path[->,>=latex, shift={(8,0)}] (x) edge (G1) (x) edge (G1.-30) (x) edge (G1.30);
        \draw[<->,>=latex]  (t1) .. controls (c1) and (c2) .. (x);
    \end{pgfonlayer}

  \end{tikzpicture}
\label{fig:kirisG7}
\end{center}
\caption{Temporal graph constructed from an instance $(G,I)$ of $2$-\textsc{WDP}, where $I=\{(s_1,t_1),(s_2,t_2)\}$ and $W=\{s_1,t_1,s_2,t_2\}$. Edges arriving in $t_1$ and $t_2$ originally from $G$ are omitted.}
\label{fig:2WDP}
\end{figure}

Define $R_1 = \{(s_1,1), (y,3)\}$ and $R_2 = \{(s_2, 3), (x,1)\}$. Now, we prove that $(G,I)$ is a ``yes'' instance of $2$-\textsc{WDP} if and only if $\mathcal{G}$ contains two edge-disjoint temporal-spanning branchings rooted at $R_1$ and $R_2$, respectively. 
Notice that snapshot $2$ of $\mathcal{G}$ is empty, thus each path in $G$ can be represented by either a temporal path on snapshot~$1$ or a temporal path on snapshot~$2$.

First, let $P_1$ and $P_2$ be two edge-disjoint paths from $s_1$ to $t_1$ and from $s_2$ to $t_2$ in $G$, respectively. 
Let $T_1$ be initially the copy of $P_1$ in snapshot~1, and $T_2$ be initially the copy of $P_2$ in snapshot~3. 
Note that the vertices not spanned by $T_1$ are all the copies of $v\notin V(P_1)$ in snapshot~1, together with all the vertices in snapshot~3, and vertices $\{(x,1),(y,3)\}$. To span snapshot~3, add to $T_1$ all edges between $(y,3)$ and $(v,3)$, for every $v\in V(G)\setminus \{s_1,t_1\}$.
To span the remainder of snapshot~1, add all edges between $(t_1,1)$ and $(v,1)$, for every $v\in V(G)\setminus (V(P_1)\cup \{s_2,t_2\})$, and the edge from $(t_1,1)$ to $(x,1)$. A similar argument can be applied to span every temporal vertex  also with $T_2$. Because $P_1$ and $P_2$ are edge-disjoint, we get that $T_1$ and $T_2$ could only intersect in the added edges, which does not occur because all edges added to $T_1$ are incident to $t_1$ and $y$, all edges added to $T_2$ are incident to $t_2$ and $x$, and there is no intersection between these.

Now, let $T_1$ and $T_2$ be edge-disjoint temporal-spanning branchings in $\mathcal{G}$ with roots $R_1,R_2$.
Denote snapshot~1 by $G_1$. Since $t_1$ appears only in $G_1$, and the only root of $R_1$ in $G_1$ is $(s_1,1)$, we get that in $T_1$ there exists a path of $G_1$ going from $(s_1,1)$ to $(t_1,1)$. Because the only incoming edge to $(x,1)$ is $(t_1,1)(x,1)$, we get that $(x,1)$ cannot be an internal vertex in this path, and hence it corresponds to a path in $G$, $P_1$. Applying a similar argument, we get a path $P_2$ from $s_2$ to $t_2$ in $G$ taken from $T_2$, and since $T_1$ and $T_2$ are edge-disjoint, so are $P_1$ and $P_2$.
\end{proof}

We conclude the proof of Item~\ref{item:two} of Theorem~\ref{thm:main} by showing the result.

\begin{theorem}\label{thm:theother}
Let $k\ge 2$ be a fixed integer, ${\cal G} = (G,\gamma,\lambda)$ be a temporal digraph, and $R_1,\ldots,R_k\subseteq V_T$. Deciding whether ${\cal G}$ has $k$ edge-disjoint temporal-spanning branchings rooted at $R_1,\cdots,R_k$ is $\NP$-complete, even if $G$ is a DAG, each snapshot has constant size, and the underlying graph of $G$ is a star. Furthermore, in this case, there is no algorithm running in time $O^*(2^{o(\T)})$ to solve the problem, unless ETH fails.
\end{theorem}
\begin{proof}
To prove \NP-completeness when the underlying graph is a star, we make a reduction from $\NAESAT$.

Let $\phi$ be a CNF formula on variables $\{x_1,\ldots,x_n\}$ and clauses $\{c_1,\ldots,c_m\}$. 
The first $n$ odd snapshots are related to variables, and the latter $m$ odd snapshots are related to clauses. The even snapshots are empty. The snapshot related to variable $x_i$ consists of the following vertices and edges:
$$V_i = \{x_i,\overline{x}_i, T\}\mbox{,\ and}$$ $$E_i=\{x_iT,\overline{x}_iT\}.$$

The snapshot related to clause $c_i = (x_{i_1}\vee x_{i_2}\vee x_{i_3})$ consists of the following vertices and edges: $$C_i=\{x_{i_1},x_{i_2},x_{i_3},T\}\mbox{,\ and}$$ 
$$E'_i = \{x_{i_1}T, x_{i_2}T, x_{i_3}T\}.$$

Formally, the temporal digraph ${\cal G}$ has lifetime $\T = 2(n+m)-1$. For each $i\in [n]$, let snapshot $G_{2i-1}$ be equal to the variable gadget related to $x_i$, and for each $i\in [m]$, let $G_{2(n+i)-1}$ be equal to the clause gadget related to clause $c_i$. We use this temporal graph to prove $\NP$-completeness.

So, let $\phi$ be a CNF formula on variables $\{x_1,\ldots,x_n\}$ and clauses $\{c_1,\ldots,c_m\}$, and let ${\cal G}=(G,\gamma, \lambda)$ be obtained as before. The fact that $G$ is a DAG, each snapshot has constant size, and the underlying graph of $G$ is a star can be easily checked. Also, the lower bound on the complexity of an algorithm to solve the problem follows because the size of ${\cal G}$ is $O(\T) = O(m+n)$. It remains to prove that $\phi$ is a ``yes'' instance if and only if ${\cal G}$ has two edge-disjoint temporal-spanning branchings ${\cal B}_1$ and ${\cal B}_2$ rooted at $R_1=R_2=\bigcup_{i=1}^{n+m}\{(x_i,2i-1),(\overline{x}_i,2i-1)\}$. The theorem holds for bigger values of $k$ by Proposition~\ref{prop:ktok+1}.

Suppose that $\phi$ is a ``yes'' instance of \NAESAT. For each true variable $x_i$, put edges $\{(x_i,2i-1)(T,2i-1)\}$ in ${\cal B}_1$, and $\{(\overline{x}_i,2i-1)(T,2i-1)\}$ in ${\cal B}_2$. For each false variable, do as before, switching between ${\cal B}_1$ and ${\cal B}_2$. Figure~\ref{fig:coloredVargadget} represents ${\cal B}_1$ with green and ${\cal B}_2$ with red, when $x_i$ is true. Now, consider a clause $c_i=(x_{i_1}\vee x_{i_2}\vee x_{i_3})$. Because this is a NAE assignment, we know that there is at least one true variable, say $x_{i_1}$, and one false variable, say $x_{i_2}$. This means that temporal edges related to $x_{i_1}T$ can only be in ${\cal B}_1$, and edges related to $x_{i_2}T$ can only be in ${\cal B}_2$. Therefore, we can span $(T,2(n-i)-1)$. Figure~\ref{fig:coloredClausegadget} represents the branching related to the assignment $(T,F,T)$ for $(x_{i_1},x_{i_2},x_{i_3})$.

\noindent\begin{minipage}{\textwidth}
\begin{minipage}{0.49\textwidth}
\begin{figure}[H]
\scriptsize
\begin{center}
  \begin{tikzpicture}[scale=0.8]
  \pgfsetlinewidth{1pt}
  \tikzset{vertex/.style={circle, minimum size=0.2cm, fill=black!20, draw, inner sep=1pt}}

  \node[vertex,label=180:{$x_i$}] (xi) at (-2,1) {};
  \node[vertex,label=90:{$T$}] (T) at (0,0) {};
  \node[vertex,label=0:{$\overline{x}_i$}] (xib) at (2,1) {};
   
   \draw[->,green] (xi)--(T);
   \draw[->,red] (xib)--(T);
   
\end{tikzpicture}
\end{center}
\caption{Branchings in a variable gadget.}
\label{fig:coloredVargadget}
\end{figure}
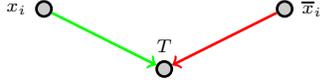\end{minipage}\hfill\begin{minipage}{0.49\textwidth}
\begin{figure}[H]
\scriptsize
\begin{center}
  \begin{tikzpicture}[scale=0.8]
  \pgfsetlinewidth{1pt}
  \tikzset{vertex/.style={circle, minimum size=0.2cm, fill=black!20, draw, inner sep=1pt}}

  \node[vertex,label=90:{$x_{i_2}$}] (x2) at (0,2) {};
  \node[vertex,label=90:{$x_{i_1}$}] (x1) at (-2,2) {};
  \node[vertex,label=90:{$x_{i_3}$}] (x3) at (2,2) {};
  \node[vertex,label=-90:{$T$}] (T) at (0,0) {};
   
   \draw[->,green] (x1)--(T) (x3)edge(T);
   \draw[->,red] (x2)--(T);
   
\end{tikzpicture}
\end{center}
\caption{Branchings in a clause gadget.}
\label{fig:coloredClausegadget}
\end{figure}
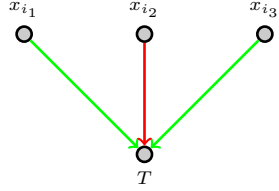
\end{minipage}
\end{minipage}

Now, let ${\cal B}_1,{\cal B}_2$ be two edge-disjoint temporal-spanning branchings rooted at $R_1=R_2=\bigcup_{i=1}^{n+m}\{(x_i,2i-1),(\overline{x}_i,2i-1)\}$, and denote by $B_i$ the set of temporal edges of ${\cal B}_i$. For each variable $x_i$, let $x_i$ be true if and only if $(x_i,2i-1)(T,2i-1)\in B_1$. Now, consider a clause $c_i=(x_{i_1}\vee x_{i_2}\vee x_{i_3})$, and denote $2(n-i)-1$ by $j$. Because $(T,j)$ must be spanned by $B_1$ and $B_2$, then at least one of the edges incoming to $(T,j)$ must be in $B_1$, say $(x_{i_1},j)(T,j)$, and at least one must be in $B_2$, say $(x_{i_2},j)(T,j)$. But this means that $(x_{i_1},2i_1-1)(T,2i_1-1)\in B_1$ (hence $x_{i_1}$ is a true variable), and that $(x_{i_2},2i_2-1)(T,2i_2-1)\in B_2$ (hence $x_{i_2}$ is a false variable), as we wanted to show.
\end{proof}

The following theorem gives us a situation where the problem becomes easy. Note that this case includes the temporal graphs used in~\cite{KKK00,MMS.19,ACGKS.19,XFJ03,SQFCA11}, where vertices are assumed to be eternal. It also implies that the problem is polynomial when the lifetime of ${\cal G}$ is $2$, which together with Theorem~\ref{thm:edgeDisjointTemporal}, gives a complete dichotomy in terms of the lifetime.

\begin{theorem}\label{thm:eternalVertices}
Let ${\cal G} = (G,\gamma,\lambda)$ be a temporal digraph with temporal vertices $V_T$, and let $R_1,\cdots,R_k\subseteq V_T$. If for every $v\in V(G)$, $\gamma(v)$ is exactly one interval of consecutive integers, then finding $k$ edge-disjoint temporal-spanning branchings rooted at $R_1,\cdots,R_k$ can be done in polynomial time.
\end{theorem}
\begin{proof}
Let $\T$ be the lifetime of ${\cal G}$. We first construct digraphs $G_0,\cdots,G_T$ and subsets $R^j_1,\cdots,R^j_k$ for each $j\in\{0,\cdots,\T\}$, then we prove that ${\cal G}$ has the desired branchings if and only if $G_j$ has $k$ edge-disjoint branchings rooted at $R^j_1,\ldots,R^j_k$ for each $j\in\{0,\cdots,\T\}$, which can be checked in polynomial time, applying Edmonds' result~\cite{edmonds1973edge}.

First, let $G_0 = (V_0,E_0)$ be the digraph in timestamp 0, i.e, $V_0 = \{u\in V(G)\mid 0\in \gamma(u)\}$ and $E_0 = \{e\in E(G)\mid (0,0)\in \gamma(e)\}$. Also, for every $i\in [k]$, let $R^0_i$ be the roots at timestamp 0, i.e., the set $\{u\in V(G)\mid (u,0)\in R_i\}$. Now, for each $j\in [\T]$, let $G_j = (V_j,E_j)$ be the digraph containing the edges arriving at timestamp $j$ together with its endpoints; more formally, $E_j = \{e\in E(G)\mid (t,j)\in \lambda(e)\mbox{, for some $t$}\}$ and $V_j = \{u\in V(G)\mid (u,j)\in V_T\mbox{ or }uv\in E_j\mbox{, for some $v$}\}$. Also, for each $i\in [k]$, let $R^j_i$ be the set of roots at times stamp $j$ together with vertices still active from the previous timestamp, i.e., $R^j_i=\{u\in V(G)\mid (u,j)\in R_i\}\cup \{u\in V(G)\mid \{i-1,i\}\subseteq \gamma(u)\}$.

Now, let ${\cal B}_1,\cdots,{\cal B}_k$ be edge-disjoint temporal-spanning branchings rooted at $R_1,\cdots,R_k$; denote by $E_T({\cal B}_i)$ the set of temporal edges of ${\cal B}_i$. Consider $j\in \{0,\cdots,\T\}$, and for each $i\in [k]$, let $B^j_i$ be the set of edges of ${\cal B}_i$ that has a copy ending at timestamp $j$, i.e., $B^j_i = \{uv\in E(G)\mid (u,h)(v,j)\in E_T({\cal B}_i)\mbox{ for some $h$}\}$. Because ${\cal B}_1,\cdots,{\cal B}_k$ are edge-disjoint, we get that $B^j_1,\cdots,B^j_k$ are also disjoint. It remains to prove that each $B^j_i$ is the edge set of a spanning branching of $G_j$ rooted at $R^j_i$. So, consider any $i\in [k]$. Because ${\cal B}_i$ is a temporal-spanning branching of ${\cal G}$, we know that each $u\in V(G)$ is either the head of some edge in $B^j_i$, in which case $u$ is spanned by $B^j_i$, or $u$ is a root in $B^j_i$. We prove that in the latter case we get that $u\in R^j_i$. Because $u$ is not the head of any edge in $B^j_i$, this means that either $(u,j)\in R_i$ or $(u,j)$ is spanned by ${\cal B}_i$ just by waiting, i.e., $\{j-1,j\}\subseteq \gamma(u)$. In both cases, we get that $u\in R^j_i$, as we wanted to prove.

Now, for each $j\in \{0,\cdots,\T\}$, let $B^j_1,\ldots,B^j_k$ be the edge sets of $k$ edge-disjoint spanning branchings of $G_j$. 
First, we prove that if $uv\in B^j_i$, then $v\in R^{j'}_{i'}$ for every $i'\in [k]$ and every $j'\in \{j+1,\cdots,\T\}\cap \gamma(v)$; hence if $B_i = \bigcup_{j=0}^{\T} B^j_i$, then we get that $B_1,\cdots,B_k$ are disjoint (these will be used later to construct the desired temporal branchings). So let $j'\in \{j+1,\cdots,k\}\cap \gamma(v)$ and observe that if $uv\in E(G_j)$ then $j\in \gamma(v)$. Because $\gamma(v)$ is an interval of consecutive integers and $j<j'\in \gamma(v)$, we get that $j'-1\in \gamma(v)$, which implies that $v\in R^{j'}_{i'}$ for every $i'\in [k]$, as we wanted to show. 
Now, for each $i\in [k]$, let ${\cal B}_i = (G,\gamma,\lambda^i)$ be a spanning temporal subdigraph of ${\cal G}$ having as temporal edges the temporal copies of each $e\in B_i$, i.e, $\lambda^i(e)=\lambda(e)$ if $e\in B_i$, and $\lambda^i(e) = \emptyset$ otherwise. Because $B_1,\cdots, B_k$ are disjoint, it follows that ${\cal B}_1,\cdots,{\cal B}_k$ are edge-disjoint, so it remains to prove that each ${\cal B}_i$ is a temporal-spanning branching rooted at $R_i$. Let $u\in V(G)$, and recall that $\gamma(u)$ is an interval of consecutive integers; denote by $s_u$ the minimum value in $\gamma(u)$. Note that we just need to prove that if $(u,s_u)\notin R_i$, then there exists a temporal edge in ${\cal B}_i$ arriving in $(u,s_u)$; this is because the other copies can be spanned simply by waiting in the interval $\gamma(u)$. Since $(u,s_u)\notin R_i$ and $s_u-1\notin \gamma(u)$, we get that $u\notin R^{s_u}_i$. So, let $vu\in B^{s_u}_i$ (it exists since $B^{s_u}_i$ is the edge set of a spanning branching of $G_{s_u}$), and recall that $\lambda^i(vu) = \lambda(vu)$. We know that $vu\in E(G_{s_u})$ only if $(v,j)(u,s_u)$ is a temporal edge of ${\cal G}$ for some $j\le s_u$ (i.e. $(j,s_u)\in \lambda(vu)$). This means that there is a temporal edge arriving in $(u,s_u)$ in ${\cal B}_i$, completing the proof.
\end{proof}


\section{Vertex-spanning Branchings}\label{sec:vertexSpanning}

This section focuses on vertex-spanning branchings, studying Problem~\ref{prob:main} in the case where $X$ is vertex, i.e. we aim to find $k$ vertex-spanning $Y$-disjoint branching, with $Y\in\{\textrm{edge},\textrm{t-edge}\}$. In the following, we provide a \NP-completeness proof to prove both Item~\ref{item:three} and Item~\ref{item:four} of Theorem~\ref{thm:main}. 

We make a reduction from \NAESAT, which consists of, given a CNF formula $\phi$ such that each clause contains exactly 3 positive literals, deciding whether there is a truth assignment to $\phi$ such that each clause has at least one true and one false literal. This is equivalent to the 2-coloring 3-uniform hypergraphs problem.

Let $\phi$ be a CNF formula on variables $\{x_1,\ldots,x_n\}$ and clauses $\{c_1,\ldots,c_m\}$. A variable gadget related to $x_i$ consists of the following vertices and edges: $$V_i = \{x_i,F_i,T_i,a_i\}\mbox{, and}$$ $$E_i = \{x_iT_i,x_iF_i,T_ia_i,F_ia_i\}.$$
And a clause gadget related to clause $c_i = (x_{i_1}\vee x_{i_2}\vee x_{i_3})$ consists of the following vertices and edges 
$$C_i = \{c_i,x_{i_1},x_{i_2},x_{i_3}\}\mbox{, and}$$ $$E'_i = \{x_{i_1}c_i,x_{i_2}c_i,x_{i_3}c_i\}.$$
Now, let $G_\phi$ be the digraph formed by the union of all variable and clause gadgets, i.e., $V(G) = \bigcup_{i=1}^n V_i\cup \bigcup_{i=1}^mC_i$ and $E(G) = \bigcup_{i=1}^n E_i\cup \bigcup_{i=1}^mE'_i$. Finally, add to $G$ two new vertices, $g,r$, and add edges $\{gx_i,rx_i\}$ for every $i\in \{1,\cdots,n\}$. See Figure~\ref{fig:vertexspan} for the digraph related to $\phi = (x_1\vee x_2\vee x_3)\wedge (x_2\vee x_3\vee x_4)$. 

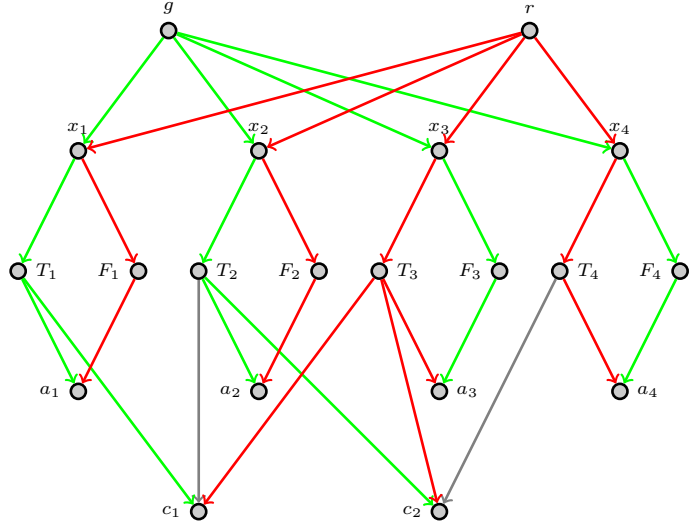
\begin{figure}[ht]
\scriptsize
\begin{center}
  \begin{tikzpicture}[scale=0.8]
  \pgfsetlinewidth{1pt}
  \tikzset{vertex/.style={circle, minimum size=0.2cm, fill=black!20, draw, inner sep=1pt}}

  \node[vertex,label={$g$}] (g) at (4.5,8) {};
  \node[vertex,label={$r$}] (r) at (10.5,8) {};
  \node[vertex,label=180:{$a_1$}] (a1) at (3,2) {};
  \node[vertex,label=180:{$a_2$}] (a2) at (6,2) {};
  \node[vertex,label=0:{$a_3$}] (a3) at (9,2) {};
   \node[vertex,label=0:{$a_4$}] (a4) at (12,2) {};
  \foreach \i in {1,...,4}{
       \node[vertex,label={$x_{\i}$}] (x\i) at (3*\i,6) {};
       \node[vertex,label={[label distance=0.005cm]0:$T_\i$}] (T\i) at (3*\i-1,4) {};
       \node[vertex,label={[label distance=0.005cm]180:$F_\i$}] (F\i) at (3*\i+1,4) {};

   }
   
   \node[vertex,label=180:{$c_1$}] (c1) at (5,0) {};
   \node[vertex,label=180:{$c_2$}] (c2) at (9,0) {};
   
   \draw[->,black!50] (T2)edge(c1) (T4)edge(c2);
   
   \draw[->,green] (g)edge(x1) (x1)edge(T1) (T1)edge(a1) (T1)edge(c1);
   \draw[->,green] (g)edge(x2) (x2)edge(T2) (T2)edge(a2)  (T2)edge(c2);
   \draw[->,green] (g)edge(x3) (x3)edge(F3) (F3)edge(a3) ;
   \draw[->,green] (g)edge(x4) (x4)edge(F4) (F4)edge(a4) ;
   
   \draw[->,red] (r) edge (x3) (x3)edge(T3) (T3)edge(a3) (T3)edge(c1);
   \draw[->,red] (r)edge(x4) (x4)edge(T4) (T4)edge(a4) (T3)edge(c2);
   \draw[->,red] (r)edge(x1) (x1)edge(F1) (F1)edge(a1) ;
   \draw[->,red] (r)edge(x2) (x2)edge(F2) (F2)edge(a2) ;
   
\end{tikzpicture}
\end{center}
\caption{Snapshot 1 related to formula $\phi = (x_1\vee x_2\vee x_3)\wedge (x_2\vee x_3\vee x_4)$, and branchings related to the assignment $(T,T,F,F)$ to $(x_1,x_2,x_3,x_4)$.}
\label{fig:vertexspan}
\end{figure}

Now, let ${\cal G}$ be the temporal digraph with lifetime 2, and such that the first snapshot is equal to $G_\phi$, while the second snapshot contains only $\{g,r\}$ and $A = \{T_i,F_i\mid i\in [n]\}$, and every edge going from $\{g,r\}$ to $A$.

\begin{theorem}\label{thm:vertexSpanning}
For each $k\ge 2$, given a temporal digraph ${\cal G}=(G,\gamma,\lambda)$ with lifetime $\T$, and set of temporal vertices $V_T$, and subsets $R_1,\cdots,R_k\subseteq V_T$, it is $\NP$-complete to decide whether ${\cal G}$ has $k$ (t-edge-disjoint or edge-disjoint) vertex-spanning branchings rooted at $R_1,\cdots,R_k$, even if $\T=2$ and $G$ is a DAG. Furthemore, letting $n=|V(G)|$ and $m=|E(G)|$, no algorithm running in time $O^*(2^{o(n+m)})$ can exist for the problem, unless ETH fails.
\end{theorem}
\begin{proof}
We prove that the problem is $\NP$-complete when $k=2$, and $\NP$-complete\-ness for bigger values of $k$ follows by Lemma~\ref{prop:ktok+1}. Let $\phi$ be an instance of \NAESAT, and let ${\cal G}$ be the temporal digraph constructed as before; denote by $G$ the base digraph. We prove that $\phi$ is a ``yes'' instance if and only if ${\cal G}$ has $k$ edge-disjoint vertex-spanning branchings rooted at $\{(g,1),(r,1)\}$ (we will see that the branchings are also t-edge disjoint).

First, suppose that $\phi$ is a ``yes'' instance of \NAESAT. We construct a green and a red branching that satisfy our conditions. For each true variable $x_i$, add to the green branching the following edges of snapshot 1:  $\{gx_i,x_iT_i,T_ia_i\}$, toghether with edge $T_ic_j$ for each clause $c_j$ containing $x_i$ that is not reached by the green branching yet; also add to the red branching edges $\{rx_i,x_iF_i,F_ia_i\}$ of snapshot 1. Do something similar to the false variables, but switching the branchings. 
Figure~\ref{fig:vertexspan} gives the branchings related to the assignment $(T,T,F,F)$ to $(x_1,x_2,x_3,x_4)$, respectively.

Observe that every $u\in V(G)$ is spanned by both branchings, with the exception of vertices in $B=\{(T_i,2),(F_i,2)\mid i\in [n]\}$. However, these can easily be spanned in the second snapshot since $\{(g,2),(r,2)\}$ is complete to $B$.

Now, let ${\cal B}_1,{\cal B}_2$ be two edge-disjoint vertex-spanning branchings. Because each $a_i$ can only be reached at the first snapshot, it is reached by exactly two paths from $\{(g,1),(r,1)\}$, one of them going through $(x_i,1)(T_i,1)$ and the other through $(x_i,1)(F_i,1)$. We then put $x_i$ as true if and only if $(x_i,1)(T_i,1)$ is in branching ${\cal B}_1$. Now, consider clause $c_i = (x_{i_1}\vee x_{i_2}\vee x_{i_3})$. Since all the paths from the roots to $c_i$ go through the temporal edges $(x_{i_1},1)(T_{i_1},1)$, $(x_{i_2},1)(T_{i_2},1)$, and $(x_{i_3},1)(T_{i_3},1)$, and $c_i$ must be spanned by ${\cal B}_1$ and ${\cal B}_2$, we get that at least one of these edges is in ${\cal B}_1$, and at least one in ${\cal B}_2$, which implies that at least one of $x_{i_1},x_{i_2},x_{i_3}$ is true, and at least one is false, as desired.

Finally, note that the reduction produces an instance of size $O(n+m)$, and the second part of the theorem follows.
\end{proof}




\section{Conclusions and Open problems}
\label{sec:conclusions}

In this paper we have investigated the temporal version of Edmonds' classical result about the problem of finding $k$ edge-disjoint spanning branchings rooted at given $R_1,\cdots,R_k$. We have introduced different definitions of spanning branchings, and of edge-disjointness in temporal digraphs. 
We have proved that, unlike the static case, only one of the these can be computed in polynomial time, namely the temporal-edge-disjoint temporal-spanning branchings problem, while the other versions are $\NP$-complete under very strict constraints. Given a temporal digraph ${\cal G}=(G,\gamma,\lambda)$, in the particular case of edge-disjoint temporal-spanning, we give separate $\NP$-complete results for fixed lifetime and fixed treewidth. A good question then might be whether there exists a polynomial algorigthm for fixed lifetime and treewidth. Another interesting question is whether the problem remains hard for fixed lifetime when the base digraph is a DAG. 
Also, as we have provided computational lower bounds under ETH in Theorem~\ref{thm:theother} and in Theorem~\ref{thm:vertexSpanning}, we wonder whether there exist algorithms matching these lower bounds.

\bibliographystyle{plain}
\bibliography{references}





\end{document}